\documentclass[preprint,12pt]{elsarticle}

\usepackage{amssymb}
\usepackage{amsthm}

\usepackage{amsfonts}
\usepackage{amssymb}
\usepackage{amsmath}

\newtheorem{thm}{Theorem}[section]

\newtheorem{lemma}[thm]{Lemma}

\newtheorem{definition}[thm]{Definition}
\newtheorem{proposition}[thm]{Proposition}

\newtheorem{conjecture}[thm]{Conjecture}

\DeclareMathOperator{\sgn}{sgn}

\usepackage{mathrsfs}

\allowdisplaybreaks

\journal{Discrete Applied Mathematics}

\begin{document}

\begin{frontmatter}



\title{On a recursive construction of circular paths and the search for $\pi$ on the integer lattice $\mathbb{Z}^2$}


\author{Michelle Rudolph-Lilith\corref{cor1}}
\ead{rudolph@unic.cnrs-gif.fr}

\address{Unit\'e de Neurosciences, Information et Complexit\'e (UNIC) \\
CNRS, 1 Ave de la Terrasse, 91198 Gif-sur-Yvette, France}


\begin{abstract}
Digital circles not only play an important role in various technological settings, but also provide a lively playground for more fundamental number-theoretical questions. In this paper, we present a new recursive algorithm for the construction of digital circles on the integer lattice $\mathbb{Z}^2$, which makes sole use of the signum function. By briefly elaborating on the nature of discretization of circular paths, we then find that this algorithm recovers, in a space endowed with $\ell^1$-norm, the defining constant $\pi$ of a circle in $\mathbb{R}^2$.
\end{abstract}



\end{frontmatter}





\section{Introduction}
\label{S_Intro}

\begin{figure}[t!]
\centering
\includegraphics[width=\linewidth]{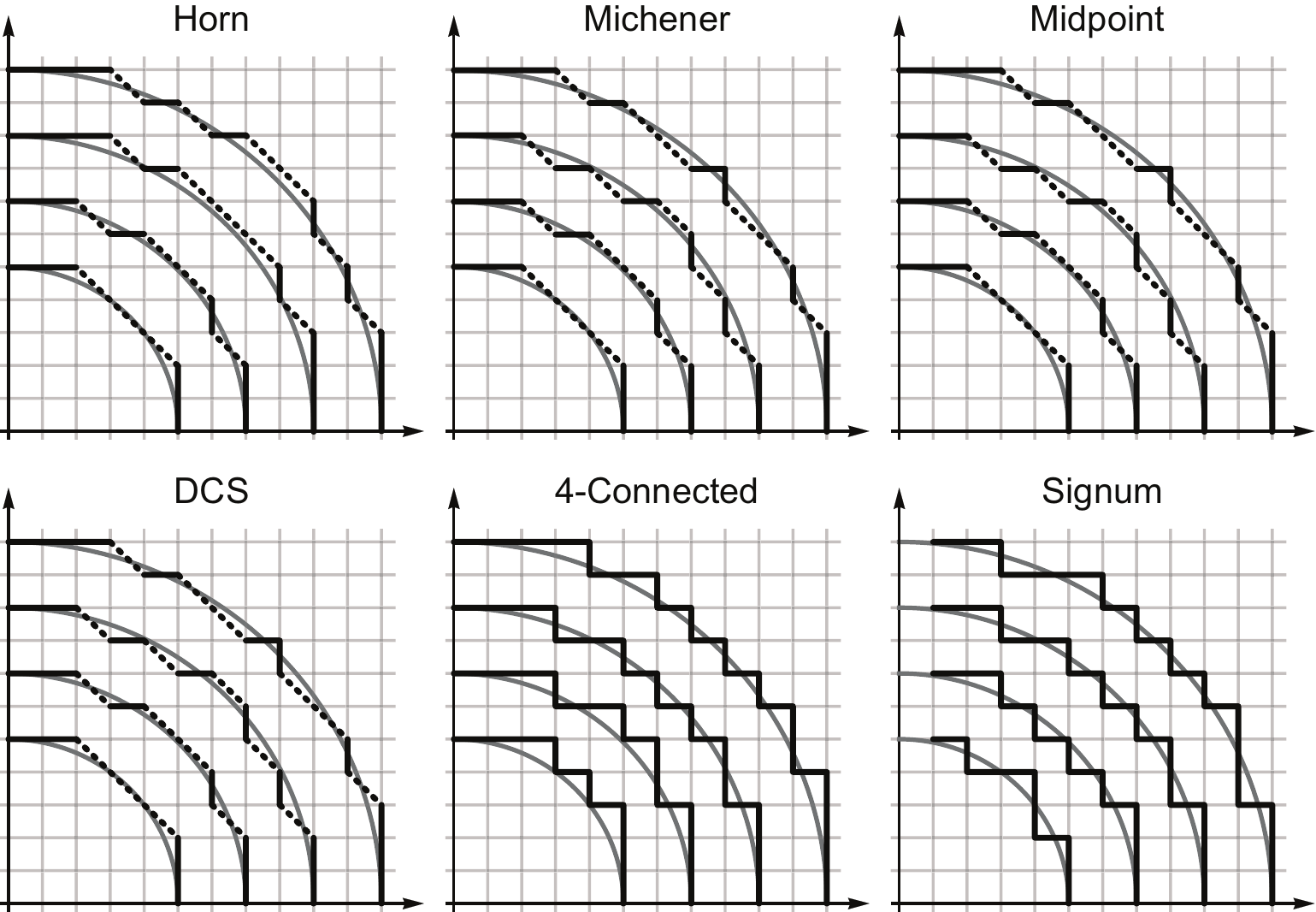}
\caption{\label{Fig_1}
Construction of digital circles using different algorithms (Horn \cite{Horn76}; Michener \cite{FoleyVanDam82}; Second-Order Midpoint \cite{FoleyEA90}; DCS \cite{BhowmickBhattacharya08}; 4-Connected \cite{BarreraEA15}; Signum: see text; for a thorough comparative study and some historical notes, see \cite{BarreraEA15}). Shown are examples of digital circles (black) approximating circles of radii 5, 7, 9 and 11 (gray). With the exception of the 4-connected and signum algorithm, most of the digital circle algorithms cited in the literature do not yield valid paths on $\mathbb{Z}^2$ (black dotted; see Definition~\ref{Def_ValidPath}).
}  
\end{figure}

The analytical characterization and algebraic representation of circles have a long history, dating back many thousands of years. With the emergence of digital computing devices utilizing grid-based interfaces in the past century, the fascination with  circles and their algorithmic generation saw another drive which significantly contributed to the evolution of discrete mathematical domains such as digital calculus and digital geometry \cite{KletteRosenfeld04, Chen14}. The interest in digital circles transcends, however, beyond application-focused paradigms. For instance, in number theory, the still unsolved Gauss's Circle Problem (e.g., see \cite{Huxley96}) or the distribution of square numbers in discrete intervals \cite{BhowmickBhattacharya08} are inherently linked to the representation of the Euclidean circle on integer lattices. In physics, a related, though perhaps controversial point is the fevered search for a quantum theory of space (and time), i.e. a discrete makeup of our world, which does ultimately lead to the rejection of the ideal real number line in favour of a discrete and finite (or effinite, see \cite{Gauthier02}) mathematical underpinning of the very construct of reality. However, despite many advances in the past decades, a rigorous and applicable framework of a discrete finite, perhaps even ultra-finite, or effinite mathematics is still largely missing, not at least due to the combinatorial complexity inherent to such approaches. 

A great number of algorithms for the generation of digital circles is known in the literature (for reviews, see \cite{Andres94, BarreraEA15}). In complexity, these algorithms range from the incremental discretization of the implicit or parametric representation of the Euclidean circle \cite{Bresenham77, Doros79, Kulpa79, McIlroy83, Kim84, NakamuraAizawa84, BiswasChaudhuri85, Pham92}, the discretization of differential equations \cite{WuRokne87, Holin91}, sophisticted spline and polygonal approximations \cite{PieglTiller89, Goldapp91, HosurMa99, BhowmickBhattacharya05}, to algorithms which utilize number-theoretical concepts \cite{BhowmickBhattacharya08}. Although all incremental algorithms utilize decision (or cost) functions, the concrete form of the latter, as well as their specific implementation, can lead to quite different representations of digital circles with the same radius (Fig.~\ref{Fig_1}). Moreover, with the exception of the 4-connected algorithm \cite{BarreraEA15} and the signum algorithm presented here, most of the used digital circle algorithms do not yield valid circular paths on the underlying 2-dimensional integer lattice. Here, a valid path is defined by

\begin{definition}
\label{Def_ValidPath}
Denoting with $\boldsymbol{x} = (x,y) \in \mathbb{Z}^2$ a point on the 2-dimensional integer lattice, a valid path $\mathcal{P}$ is defined as a set of points $\{ \boldsymbol{x}_n \}$ such that $\forall \boldsymbol{x}_n \in \mathcal{P}$, there exist at most two $\boldsymbol{x}_m , \boldsymbol{x}_{m'} \in \mathcal{P}$ with $m \neq m' \neq n$ such that $\lVert \boldsymbol{x}_n - \boldsymbol{x}_m \rVert_1 = 1$ and $\lVert \boldsymbol{x}_n - \boldsymbol{x}_{m'} \rVert_1 = 1$, where $\lVert \boldsymbol{x} \rVert_1 = |x|+|y|$ denotes the $\ell^1$-norm on $\mathbb{Z}^2$. For a valid closed path, there exist, for each $\boldsymbol{x}_n$, exactly two such $\boldsymbol{x}_m , \boldsymbol{x}_{m'} \in \mathcal{P}$ with the aforementioned properties.
\end{definition}

In this paper, we will present a simple recursive algorithm, the signum algorithm, which generates a valid circular path on a 2-dimensional integer lattice $\mathbb{Z}^2$ (Section \ref{S_SignumAlgorithm}). Although this algorithm can not be viewed as the computationally most efficient digital circle algorithm, it allows for easy generalization to higher dimensions, thus providing a viable algorithm for constructing spheres and, generally, hyperspheres of integer radii in $\mathbb{R}^3$ and $\mathbb{R}^n$, respectively. In Section 3, we then briefly elaborate on the discretization of circles in $\mathbb{R}^2$, and present some  findings which show that the numerical value of $\pi$ can be recovered in the asymptotic limit using solely the Manhattan distance ($\ell^1$-norm), thus providing an interesting link between Euclidean geometry and geometrical constructions on $\mathbb{Z}^2$. 


\section{The signum algorithm}
\label{S_SignumAlgorithm}

In order to construct a valid path on $\mathbb{Z}^2$ which approximates a circle of integer radius $r$ in $\mathbb{R}^2$, we follow an approach similar to that used in most of the known digital circle algorithms, namely utilizing a cost function to assign points on $\mathbb{Z}^2$ to the digital circle. For reasons of symmetry and notational simplicity, we restrict throughout the paper to constructing a quarter circle in the upper right quadrant starting from the horizontal axis, and assume the origin of the circle $\boldsymbol{o} = (0,0)$. 


\subsection{Recursive construction of a valid circular path on $\mathbb{Z}^2$}
\label{SS_Recursion}

\begin{figure}[t!]
\centering
\includegraphics[width=\linewidth]{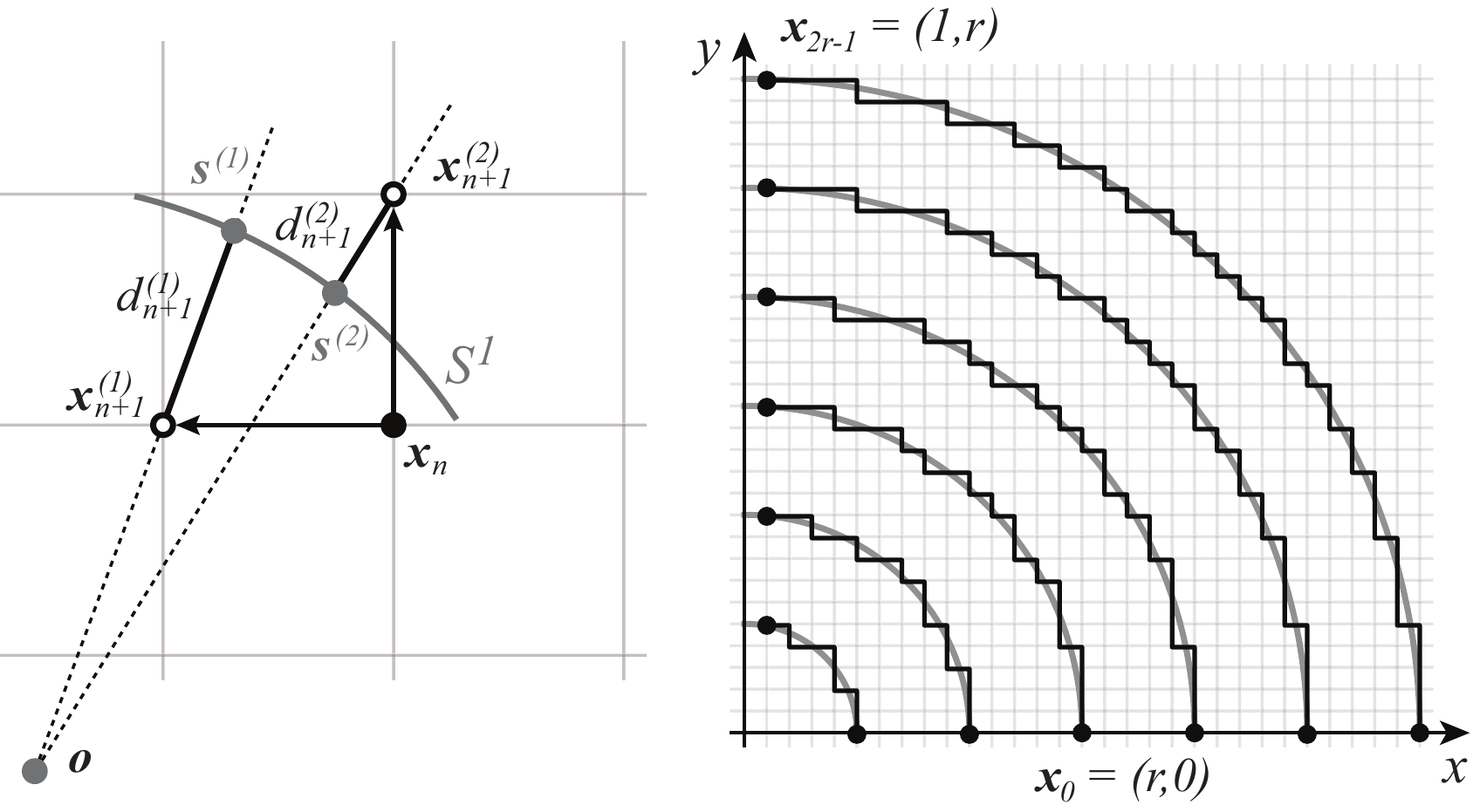}
\caption{\label{Fig_2}
Recursive construction of a circular path on $\mathbb{Z}^2$ in the upper right quadrant, approximating $S^1 \subset \mathbb{R}^2$ (left; see text for explanation), and examples of digital circles of various integer radii ($r=5,10,15,20,25,30$) constructed by using the signum algorithm (right).
}  
\end{figure}

Let $\mathcal{S}^1$ denote a circular path on $\mathbb{Z}^2$, and $S^1$ a circle on $\mathbb{R}^2$. Given $\boldsymbol{x}_n = (x_n,y_n) \in \mathcal{S}^1$ with $x_n, y_n \in \mathbb{Z}, n \in \mathbb{N}$, there are only two possibilities for the unassigned neighbouring point $\boldsymbol{x}_{n+1}$ along the circular path (Fig.~\ref{Fig_2}, left), namely
\begin{equation}
\label{Eq_xn1}
\boldsymbol{x}_{n+1} = (x_{n+1},y_{n+1}) =
\left\{ 
\begin{array}{l}
\boldsymbol{x}_{n+1}^{(1)} = (x_n-1,y_n), \\[0.2em]
\text{or, } \boldsymbol{x}_{n+1}^{(2)} = (x_n,y_n+1).
\end{array}
\right.
\end{equation}
In order to decide between $\boldsymbol{x}_{n+1}^{(1)}$ and $\boldsymbol{x}_{n+1}^{(2)}$, we utilize a cost function based on a minimum criterion. To that end, consider the intersections $\boldsymbol{s}^{(1)}$, $\boldsymbol{s}^{(2)}$ on $S^1$ of lines through $\boldsymbol{o}$ and $\boldsymbol{x}_{n+1}^{(1)}$, $\boldsymbol{x}_{n+1}^{(2)}$, respectively. The line segments $\overline{\boldsymbol{s}^{(1)} \boldsymbol{x}_{n+1}^{(1)}}$ and $\overline{\boldsymbol{s}^{(2)} \boldsymbol{x}_{n+1}^{(2)}}$ have a respective Euclidean length of
\begin{equation}
\label{Eq_dn11}
d_{n+1}^{(1)} 
= \left| r - \lVert \boldsymbol{x}_{n+1}^{(1)} \rVert_2 \right| 
= \left| r - \sqrt{(x_n-1)^2 + y_n^2} \right|
\end{equation}
and
\begin{equation}
\label{Eq_dn12}
d_{n+1}^{(2)} 
= \left| r - \lVert \boldsymbol{x}_{n+1}^{(2)} \rVert_2 \right| 
= \left| r - \sqrt{x_n^2 + (y_n+1)^2} \right| ,
\end{equation}
where $\lVert \boldsymbol{x} \rVert_2 = \sqrt{x^2+y^2}$ denotes the $\ell^2$-norm (Euclidean norm) in $\mathbb{R}^2$. With this, the minimization criterion is then given by
\begin{equation}
\label{Eq_xn1a}
\boldsymbol{x}_{n+1} = 
\left\{ 
\begin{array}{ll}
\boldsymbol{x}_{n+1}^{(1)} & \text{ if } d_{n+1}^{(1)} \leq d_{n+1}^{(2)} \\[0.2em]
\boldsymbol{x}_{n+1}^{(2)} & \text{ if } d_{n+1}^{(1)} > d_{n+1}^{(2)} .
\end{array}
\right.
\end{equation}
We note that the equal sign in the case $\boldsymbol{x}_{n+1} = \boldsymbol{x}_{n+1}^{(1)}$ is convention to account for the unlikely scenario that $d_{n+1}^{(1)} = d_{n+1}^{(2)}$. If $d_{n+1}^{(1)}$ and $d_{n+1}^{(2)}$ are equal, both $\boldsymbol{x}_{n+1}^{(1)}$ and $\boldsymbol{x}_{n+1}^{(2)}$ are equally valid neighbours of $\boldsymbol{x}_{n}$, and we choose, without loss of generality, $\boldsymbol{x}_{n+1}^{(1)}$.

To construct the associated cost function, we define
\begin{equation}
s_n := \sgn(\Delta_n)
\end{equation}
with 
\begin{equation}
\Delta_n := d_{n+1}^{(1)} - d_{n+1}^{(2)}
= \left| r - \sqrt{(x_n-1)^2 + y_n^2} \right| - \left| r - \sqrt{x_n^2 + (y_n+1)^2} \right|,
\end{equation}
and
\begin{equation}
\label{Eq_Sgn}
\sgn(x) =
\left\{
\begin{array}{ll}
-1 & \text{ if } x \leq 0 \\
1  & \text{ if } x > 0
\end{array}
\right.
\end{equation}
denoting the signum function. Please note that (\ref{Eq_Sgn}) slightly deviates from the commonly used notion of the signum function in that it assigns to $x=0$ a value $\sgn(0) = -1$ instead of $\sgn(0) = 0$. This redefinition allows to accommodate the unlikely case $d_{n+1}^{(1)} = d_{n+1}^{(2)}$ in (\ref{Eq_xn1a}), and, again, does not lead to loss of generality. With this, (\ref{Eq_xn1a}) takes the form
\begin{equation}
\label{Eq_xn1b}
\boldsymbol{x}_{n+1} = (x_{n+1},y_{n+1})
\left\{ 
\begin{array}{ll}
\boldsymbol{x}_{n+1}^{(1)} = (x_n-1, y_n) & \text{ if } s_n = -1 \\[0.2em]
\boldsymbol{x}_{n+1}^{(2)} = (x_n, y_n+1) & \text{ if } s_n = 1 .
\end{array}
\right.
\end{equation}
Utilizing the signum function (\ref{Eq_Sgn}), we can then rewrite (\ref{Eq_xn1b}) in algebraic form as
\begin{equation}
\label{Eq_xn1c}
\left\{
\begin{array}{l}
x_{n+1} = \frac{1}{2} ( 1 - s_n ) ( x_n - 1 ) + \frac{1}{2} ( 1 + s_n ) x_n \\[0.2em]
y_{n+1} = \frac{1}{2} ( 1 - s_n ) y_n + \frac{1}{2} ( 1 + s_n ) ( y_n + 1 ) .
\end{array}
\right.
\end{equation}

We observe that, by construction, the circular path $\mathcal{S}^1$ intersects in the considered upper right quadrant with the horizontal and vertical axis at $(r,0)$ and $(0,r)$, respectively. As the Manhattan distance between these two intersection points counts the number of points on $\mathbb{Z}^2$ along a valid circular path $\mathcal{S}^1$, each quadrant will contribute $2r$ points to $\mathcal{S}^1$. With this, after simplification of (\ref{Eq_xn1c}), we can then formulate the following

\begin{proposition}
\label{Prop_SignumAlgorithm}
A valid circular path $\mathcal{S}^1 \subset \mathbb{Z}^2$ approximating a circle $S^1 \subset \mathbb{R}^2$ with radius $r \in \mathbb{N}$ and origin $\boldsymbol{o} = (0,0)$ in the upper right quadrant is a set $\{ \boldsymbol{x}_n \}$ of $2r$ points $\boldsymbol{x}_n = (x_n,y_n)$ with $x_n, y_n \in \mathbb{Z}$ obeying the algebraic recursions
\begin{equation}
\label{Eq_S1Algorithm}
\left\{
\begin{array}{l}
x_0 = r , x_{n+1} = x_n + \frac{1}{2} s_n - \frac{1}{2} \\[0.2em]
y_0 = 0 , y_{n+1} = y_n + \frac{1}{2} s_n + \frac{1}{2} ,
\end{array}
\right.
\end{equation}
where $n \in [0,2r-1], n \in \mathbb{N}$, and 
\begin{equation}
\label{Eq_Cost}
s_n = \sgn(\Delta_n)
\end{equation}
with
\begin{equation}
\Delta_n = \left| r - \sqrt{(x_n-1)^2 + y_n^2} \right| - \left| r - \sqrt{x_n^2 + (y_n+1)^2} \right| 
\end{equation}
denoting the cost function.
\end{proposition}

As the proposed algorithm makes solely use of the signum function, we will, for notational convenience, refer to as \emph{signum algorithm} in the remainder of this paper. Furthermore, we note that 
\begin{equation*}
\Delta_0 = 1-|r-\sqrt{r^2+1}| \geq 2 - \sqrt{2} > 0,
\end{equation*}
$\forall r \geq 1$, thus $s_0 = 1$. Figure~\ref{Fig_2} (right) shows representative examples of digital circles of various integer radii, constructed using the signum algorithm.

Proposition~\ref{Prop_SignumAlgorithm} provides a recursive algorithm for constructing digital circles of integer radii on $\mathbb{Z}^2$. Starting at $\boldsymbol{x}_0 = (r,0)$, this algorithm yields $2r$ successive points forming a valid circular path in the upper right quadrant on the integer lattice. This contrasts, for instance, the most widely used Bresenham \cite{Bresenham77} and Midpoint \cite{FoleyEA90} algorithms, which deliver only about 70\% of the points necessary for a valid circular path on $\mathbb{Z}^2$ (see Fig.~\ref{Fig_1}). Moreover, in contrast to many known digital circle algorithms, the computational implementation of the signum algorithm does not require decision trees or case distinctions, but solely relies on the signum function to generate a valid path. Such an algebraic formulation has the advantage of being mathematical tractable and allowing for rigorous manipulations. Specifically, due to the special properties of the signum function $\sgn(x): \mathbb{R} \rightarrow \{-1,1\}$, the cost function (\ref{Eq_Cost}) can further be simplified, as shown in the next section. 

Finally, we note that the geometrical basis and algebraic representation of the signum algorithm allows for direct generalization to higher dimensions. Specifically, for each given $(1/2^n)^{\text{th}}$ hypersphere in $\mathbb{R}^n$ (the generalization of the quarter circle in $\mathbb{R}^2$), Eq.~(\ref{Eq_xn1}) must be extended to  encompass $n$ possible neighbours for each given point along a valid ``hypercircular path''. Generalizing the Euclidean distance of the associated line segments, Eqs.~(\ref{Eq_dn11}) and (\ref{Eq_dn12}), to $\mathbb{R}^n$ will then yield a number of minimization criteria corresponding to (\ref{Eq_xn1a}) which can be expressed by utilizing the signum function alone, and lead to a recursive algorithm constructing a $(n-1)$-dimensional hypercircular ``path'' of integer radius on the $n$-dimensional integer lattice $\mathbb{Z}^n$.


\subsection{Simplification of the cost function}
\label{SS_CostFunction}

The computational complexity of the digital circle algorithm presented in Proposition~\ref{Prop_SignumAlgorithm} is carried by the argument of the cost function, which requires to evaluate the square root of integer numbers. However, as we show below, due to the properties of the signum function, $\Delta_n$ can be significantly simplified. To that end, we first formulate 

\begin{lemma}
\label{Lemma_sgn}
The signum function $\sgn(x): \mathbb{R} \rightarrow \{-1,1\}$ with
\begin{equation}
\label{Eq_sgn}
\sgn(x) =
\left\{
\begin{array}{ll}
-1 & \text{ if } x \leq 0 \\
1  & \text{ if } x > 0
\end{array}
\right.
\end{equation}
is subject to the following property:
\begin{equation}
\label{Eq_sgn1}
\sgn(x-y) = \sgn(f(x)-f(y))
\end{equation}
for all $x,y \in \mathbb{R}: x,y \geq 0$ and strict monotonically increasing functions $f(x): \mathbb{R} \rightarrow \mathbb{R}$. Moreover, $\forall x \in \mathbb{R}: x \neq 0$ and $a \in \mathbb{R}$
\begin{equation}
\label{Eq_sgn2}
\sgn(ax) = \left\{
\begin{array}{ll}
\sgn(x)  & \text{ if } a > 0 \\
-\sgn(x) & \text{ if } a < 0.
\end{array}
\right.
\end{equation}
\end{lemma}

\begin{proof}
Eqs.~(\ref{Eq_sgn1}) and (\ref{Eq_sgn1}) are self-evident from the definition of the signum function (\ref{Eq_sgn}).
\end{proof}

Utilizing Lemma~\ref{Lemma_sgn}, we can now formulate

\begin{proposition}
The cost function $s_n$ in Proposition~\ref{Prop_SignumAlgorithm} is equivalent to
\begin{equation}
\label{Eq_CostSimplified}
s_n = -\sgn\left(
a_n + \frac{r}{\sqrt{2}} \left( \sqrt{(a_n-1)^2+c_n^2} - \sqrt{(a_n+1)^2+c_n^2} \right)
\right),
\end{equation}
where $a_n = x_n + y_n$ with $n \in [0,2r-1], n \in \mathbb{N}$ obeys the recursion
\begin{equation}
\label{Eq_anRec}
a_0 = r , a_{n+1} = a_n + s_n
\end{equation}
and $c_n = r-n-1$. Furthermore, for $r > 4$, the cost function can be approximated by
\begin{equation}
\label{Eq_CostApproximated}
s_n = -\sgn\left( a_n^2 + c_n^2 + 1 - 2r^2 \right).
\end{equation}
\end{proposition}

\begin{proof}
First we will show (\ref{Eq_CostSimplified}). To that end, we observe that $f(x)=x^2$ for $x \geq 0$ obeys the condition of Lemma~\ref{Lemma_sgn}, thus
\begin{eqnarray*}
s_n 
& = & \sgn\left( \left| r - \sqrt{(x_n-1)^2 + y_n^2} \right|^2 - \left| r - \sqrt{x_n^2 + (y_n+1)^2} \right|^2 \right) \\
& = & \sgn\left( -2(x_n+y_n) - 2r \left( \sqrt{(x_n-1)^2+y_n^2} - \sqrt{x_n^2+(y_n+1)^2} \right) \right) \\
& = & -\sgn\left( x_n+y_n + r \left( \sqrt{(x_n-1)^2+y_n^2} - \sqrt{x_n^2+(y_n+1)^2} \right) \right) \\
& = & -\sgn\Big( a_n + \frac{r}{\sqrt{2}} \Big( \sqrt{a_n^2-2a_n+b_n^2-2b_n+2} \\
&   & \hspace*{28mm} - \sqrt{a_n^2+2a_n+b_n^2-2b_n+2} \Big) \Big),
\end{eqnarray*}
where in the last two steps Eq.~(\ref{Eq_sgn2}), $a_n := x_n + y_n$ and $b_n := x_n - y_n$ were used. Observing that $b_n$ obeys the recursion 
\begin{equation*}
b_0 = r, b_{n+1} = b_n - 1,
\end{equation*}
hence takes the explicit form $b_n = r-n$, and defining further 
\begin{equation}
\label{Eq_cn}
c_n := \sqrt{b_n^2 - 2b_n + 1} = r-n-1,
\end{equation}
we arrive at Eq.~(\ref{Eq_CostSimplified}).

To show (\ref{Eq_CostApproximated}), we first note that $a_n \geq r, \forall n \in [0, 2r-1]$, with the minimum taken at $n=0$. The maximum is reached for a point on the circular path which, when connected to the origin by a line in $\mathbb{R}^2$, takes an angle with the horizontal axis closest to $\pi/4$. As $a_n = x_n + y_n$ is, in the upper right quadrant, equivalent to the Manhattan distance of $(x_n,y_n)$, we can approximate
\begin{equation*}
\max_{n} a_n \approx r \cos\left(\frac{\pi}{4}\right) + r \sin\left(\frac{\pi}{4}\right) = \sqrt{2} r.
\end{equation*}
As any point on the circular path $\mathcal{S}^1$ does, by construction, reside at most $\sqrt{2}$ away from the closest point on $S^1$, we can securely assume that $a_n \leq \sqrt{2}(r+1), \forall n \in [0, 2r-1]$. Thus,
\begin{equation*}
\begin{array}{rcccl}
r & \leq & a_n & \leq & \sqrt{2}(r+1) \\
r^2 & \leq & a_n^2 & \leq & 2(r^2+2r+1).
\end{array}
\end{equation*}
Similarly, with (\ref{Eq_cn}), $c_n$ takes its minimum of $0$ at $n=r-1$, and its maximum of $r$ for $n=2r-1$. With this, we have the following inequality
\begin{equation*}
r^2+1 \leq a_n^2+c_n^2+1 \leq 3r^2+4r+3,
\end{equation*}
from which 
\begin{equation*}
\frac{2 a_n}{a_n^2+c_n^2+1} < 1
\end{equation*}
$\forall r \geq 4$ follows. With this, we can rewrite (\ref{Eq_CostSimplified}), using again Lemma~\ref{Lemma_sgn}, and obtain
\begin{equation*}
s_n = -\sgn\Big(
a_n^2 - \frac{r^2}{2} (a_n^2+c_n^2+1) \Big(
\sqrt{1 + \tfrac{2a_n}{a_n^2+c_n^2+1}} - \sqrt{1 - \tfrac{2a_n}{a_n^2+c_n^2+1}}
\Big) \Big).
\end{equation*}
Observing that
\begin{equation*}
\left( \sqrt{1-x} - \sqrt{1+x} \right)^2
= \sum\limits_{k=1}^{\infty} \binom{2(k-1)}{k-1} \frac{4}{2^{2k} k} x^{2k} 
\end{equation*}
$\forall x \in \mathbb{R}: |x| \leq 1$, we then expand, for $r \geq 4$, the argument of $s_n$ in a power series. This yields
\begin{equation*}
s_n = -\sgn\left(
a_n^2 - \frac{r^2}{2} (a_n^2+c_n^2+1) \sum\limits_{k=1}^{\infty} \binom{2(k-1)}{k-1} \frac{4}{2^{2k} k} \left(\frac{2a_n}{a_n^2+c_n^2+1}\right)^{2k} 
\right).
\end{equation*}
For large $r$, the sum in the last equation converges rapidly, and we can approximate $s_n$ by taking only the leading term $k=1$ into consideration, thus showing (\ref{Eq_CostApproximated}).
\end{proof}

We note that, whereas (\ref{Eq_Cost}) and (\ref{Eq_CostSimplified}) provide exact expressions for the cost function $s_n$, Eq.~(\ref{Eq_CostApproximated}) provides an approximation which, for $r \gg 1$, yields the same result as the exact expressions. However, using (\ref{Eq_CostApproximated}) will significantly lower the computational cost of constructing a digital circle, as here only integer operations are involved. Finally, we remark that both the exact alternative form of the cost function (\ref{Eq_CostSimplified}) and its approximation (\ref{Eq_CostApproximated}) are no longer given in terms of the coordinates $(x_n,y_n)$ of points along the circular path $\mathcal{S}^1$, but instead are functions of the Manhattan distance $a_n = |x_n| + |y_n|$ of each point $(x_n,y_n) \in \mathcal{S}^1$ to the center of the circle. The resulting finite sequence itself is subject to a recursion, see Eq.~(\ref{Eq_anRec}), and will be used in the next section to recover the numerical value of $\pi$ from a digital circle $\mathcal{S}^1 \subset \mathbb{Z}^2$.


\section{The search for $\pi$ on $\mathbb{Z}^2$}
\label{S_PiZ2}

By construction, each digital circle algorithm delivers, for any given radius $r$, a set of points on $\mathbb{Z}^2$ which, for increasing $r$, approximates with increasing precision $S^1 \subset \mathbb{R}^2$ when each pair of nearest neighbouring points is connected with a straight line in $\mathbb{R}^2$ (see Fig.~\ref{Fig_1}), eventually yielding $S^1$ for $r \rightarrow \infty$. However, if we restrict to $\mathbb{Z}^2$ with its $\ell^1$-norm, all valid circular paths will remain finitely distinct from $S^1$ even in the asymptotic case, as each path is bound to the lattice. To make matters worse, if we consider the distance of each point along the circular path to the origin, then we find that it is no longer constant. This, although being a known characteristic with amusing consequences of geometric spaces endowed with $\ell^1$-norm \cite{Krause87}, it is in direct conflict with the very original definition of a circle as put forth in Euclid's \emph{Elements} (Book I, \S 19). If we adhere to Euclid's circle definition in such a discrete space with $\ell^1$-norm, on the other hand, the discrete circle takes, in the continuum limit, the shape of a square rotated by $\pi/4$. Thus, in other words, a \emph{digital} circle and a \emph{discrete} circle are two distinct geometrical objects. 


\subsection{Reconciling digital and discrete circles}
\label{SS_DigitalDiscreteCircles}

Digital geometry defines a ``digital circle'' simply as a discrete approximation (or digitized model) of a circle in $\mathbb{R}^2$ obtained by searching for points on $\mathbb{Z}^2$ which are closest to $S^1$. Naturally, the form of each model will carry consequences for its underlying relationship to the circle on $\mathbb{R}^2$. We can thus interrogate the geometric properties of each model in $\mathbb{Z}^2$ and $\mathbb{R}^2$, specifically, explore the relationship between properties of the digital circle $\mathcal{S}^1 \subset \mathbb{Z}^2$, i.e. a circular path in a discrete space endowed with $\ell^1$-norm, and the properties of $S^1 \subset \mathbb{R}^2$, i.e. a circle in a continuous space endowed with $\ell^2$-norm. We will focus here on the defining constant of circles, $\pi$, and show below that the parametric and polar discretizations of the circle lead to an overestimate for $\pi$, measured both numerically and analytically, whereas the signum algorithm introduced in Section~\ref{S_SignumAlgorithm} allows to recover its correct value in a somewhat surprising fashion.

Before outlining the details of this interrogation, we note that, firstly, an alternative, and mathematically more rigorous, definition of a circle in $\mathbb{R}^2$ is given by its parametric representation. Specifically, a circle $S^1 \subset \mathbb{R}^2$ is the set of all points $(x,y) \in \mathbb{R}^2$ which satisfy the algebraic relation 
\begin{equation}
\label{Eq_circleP}
x^2 + y^2 = r^2,
\end{equation} 
where $r \in \mathbb{R}: r > 0$ is called the radius of the circle. Recalling Proposition~\ref{Prop_SignumAlgorithm}, a digital circle $\mathcal{S}^1 \subset \mathbb{Z}^2$ is the set of all points $(x,y) \in \mathbb{Z}^2$ satisfying a specific recursive algebraic relation corresponding to Eq.~(\ref{Eq_S1Algorithm}) in the upper right quadrant.

Secondly, although differences exist in the mathematical representation of the algorithmic search for points on $\mathbb{Z}^2$ closest to $S^1$, each digital circle algorithm utilizes the Euclidean norm in one form or another in its minimization criterion. The same holds for the signum algorithm presented here. However, the resulting cost function (\ref{Eq_CostSimplified}) and its approximation (\ref{Eq_CostApproximated}) are given in terms of $a_n = x_n + y_n$, which corresponds, in the upper right quadrant, to the Manhattan distance of the point $(x_n,y_n) \in \mathcal{S}^1$ to the origin. Taking both arguments together, it could be contended that the ``digital circle'' constructed by the signum algorithm is not only a digital model of $S^1 \subset \mathbb{R}^2$, but a valid discrete model of a circle in $\mathbb{Z}^2$, a space endowed with $\ell^1$-norm, with properties which, in the asymptotic limit, translate into those of $S^1$.


\subsection{$\pi$ in discretized circles}
\label{SS_piND}

To illustrate this crucial latter point, we will consider the defining constant of a circle in $\mathbb{R}^2$ (or hyperspheres in $\mathbb{R}^n $ in general), namely $\pi$, and ask whether $\pi$ can be obtained in a discrete space endowed with $\ell^1$-norm. To that end, we first recall how $\pi$ is obtained on $\mathbb{R}^2$ by calculating the circumference of the circle. Given the parametric representation of $S^1 \subset \mathbb{R}^2$, Eq.~(\ref{Eq_circleP}), we have $y=\pm\sqrt{r^2-x^2}$ and for the circumference $\mathcal{C}$, using the arc length,
\begin{equation}
\label{Eq_S}
\mathcal{C} 
= 2 \int\limits_{-r}^r \sqrt{1+\left(\frac{dy}{dx}\right)^2}
= 2 \int\limits_{-r}^r \text{d}x \,
\sqrt{1+\frac{x^2}{r^2-x^2}} 
= 2 \pi r.
\end{equation}

Equation~(\ref{Eq_S}) can be viewed as a definition of $\pi$ in terms of the ratio between the circumference of a circle and the (Euclidean) distance of each point on $S^1$ to the center, i.e.
\begin{equation}
\label{Eq_pi}
\pi := \frac{\mathcal{C}}{2r}
\end{equation}
for $r > 0$. Remaining for a moment in $\mathbb{R}^2$, but replacing the Euclidean distance $r$ by the Manhattan distance $a(x,y)=|x|+|y|$ of each point $(x,y) \in S^1$ to the center, we can define
\begin{equation}
\label{Eq_pi}
\pi(x,y) := \frac{\mathcal{C}}{2 a(x,y)} = \frac{4r}{a(x,y)},
\end{equation}
where we used the fact that the circumference of a circle in a space with $\ell^1$-norm is $\mathcal{C}=8r$. As mentioned above, as $a(x,y)$ changes depending on the point along the circle (see Fig.~\ref{Fig_3}, top left), $\pi(x,y)$ will be a function of 
$(x,y) \in S^1$, with values ranging between 4 and $2\sqrt{2}$ (see Fig.~\ref{Fig_3}, top right), and the value of $\pi$ residing in between these bounds. Using the parametric representation of a circle, 
\begin{equation}
\label{Eq_xy}
\left\{
\begin{array}{l}
x = r \cos(\varphi) \\
y = r \sin(\varphi)
\end{array}
\right.
\end{equation}
with $0 \leq \varphi \leq 2\pi$, we have 
\begin{equation}
\label{Eq_aphi}
a(x,y) \equiv a(r,\varphi) = |r \cos(\varphi)| + |r \sin(\varphi)|.
\end{equation}
With this, we can calculate the average of (\ref{Eq_pi}) over all points on $S^1$ (due to symmetry, it is sufficient to restrict to the upper right quadrant), which yields
\begin{equation}
\label{Eq_picont}
\overline{\pi}
= \frac{2}{\pi} \int\limits_0^{\pi/2} \text{d}\varphi \, \frac{4}{\cos(\varphi)+\sin(\varphi)}
= \frac{8}{\pi} \sqrt{2} \text{ arctanh}\left( \frac{1}{\sqrt{2}} \right)
\sim 3.17406.
\end{equation}
Note that the obtained value is independent of $r$. More interestingly, however, is the fact that the obtained value is close, but not identical, to $\pi$.

\begin{figure}[t!]
\centering
\includegraphics[width=\linewidth]{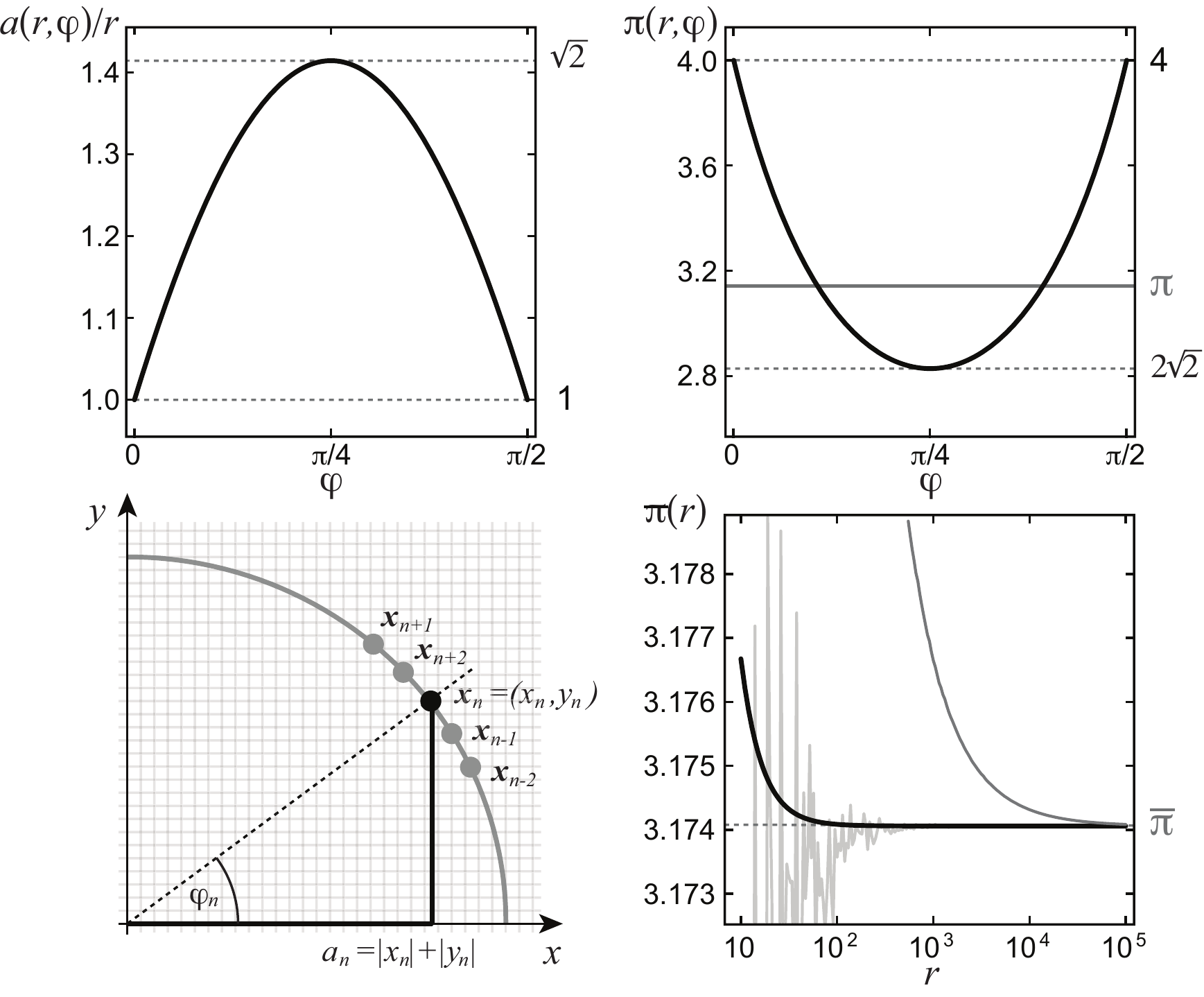}
\caption{\label{Fig_3}
Relative Manhattan distance $a(\varphi)/r$ (top left; see Eq.~(\ref{Eq_aphi})) and associated $\pi$-values (top right; see Eq.~(\ref{Eq_pi})) along points on $S^1 \subset \mathbb{R}^2$. Parametric discretization of the circle $S^1 \subset \mathbb{R}^2$ (bottom left; see text for explanation) and the resulting arithmetic mean of the $\pi_n$ values associated with each point on $S^1$ (see Eq.~(\ref{Eq_ApinND})) as function of the radius $r$ (bottom right; black: $a_n(r)$ given by Eq.~(\ref{Eq_xnyndigi}), light grey: $a_n(r)$ given by Eq.~(\ref{Eq_xnyndigiS}), dark grey: $a_n(r)$ given by Eq.~(\ref{Eq_xnyndigiR})). The asymptotic value $\overline{\pi}$ for $r \rightarrow \infty$, Eq.~(\ref{Eq_pidigi}), differs from $\pi$ in all cases.
}  
\end{figure}

The same holds true if we perform a parametric discretization of $S^1 \subset \mathbb{R}^2$ by introducing $2r$ discrete angles\begin{equation}
\label{Eq_xnyn}
\left\{
\begin{array}{l}
x_n = r \cos(\varphi_n) \\
y_n = r \sin(\varphi_n)
\end{array}
\right.
\end{equation}
with 
\begin{equation}
\varphi_n = \frac{n}{2r} \, \frac{\pi}{2},
\end{equation}
$n \in [0,2r-1], n \in \mathbb{N}$ (Fig.~\ref{Fig_3}, bottom left). In this case, remaining with the $\ell^1$-norm, we have 
\begin{equation}
\label{Eq_xnyndigi}
a(x_n,y_n) \equiv a_n(r) = |r \cos(\varphi_n)| + |r \sin(\varphi_n)|.
\end{equation}
Defining, similar to (\ref{Eq_pi}), $\pi$-values associated with each point along the now discretized circle according to
\begin{equation}
\label{Eq_piND}
\pi_n(r) := \frac{\mathcal{C}}{2 a_n(r)} = \frac{4r}{a_n(r)},
\end{equation}
we consider the arithmetic mean $A(\pi_n)$ of all $\pi_n(r)$, i.e. 
\begin{equation}
\label{Eq_ApinND}
A(\pi_n) = \frac{1}{2r} \sum\limits_{n=0}^{2r-1} \pi_n(r),
\end{equation}
and obtain
\begin{equation}
\label{Eq_ApinND1}
A(\pi_n) 
= 2 \sum\limits_{n=0}^{2r-1} \frac{1}{a_n(r)}
= \frac{2}{r} \sum\limits_{n=0}^{2r-1} \frac{1}{\cos\left(\frac{n\pi}{4r}\right)+\sin\left(\frac{n\pi}{4r}\right)}.
\end{equation}
To simplify the last equation, we first rewrite the denominator under the sum using 
\begin{equation*}
\sin(x) \pm \cos(y) = 2 \sin\left(\frac{1}{2}(x \pm y) \pm \frac{\pi}{4}\right) \cos\left(\frac{1}{2}(x \mp y) \mp \frac{\pi}{4}\right)
\end{equation*}
(\cite{GradshteynRyzhik07}, relation 1.314.9$^*$). With this, (\ref{Eq_ApinND1}) takes the form
\begin{eqnarray*}
A(\pi_n)
& = & \frac{\sqrt{2}}{r} \sum\limits_{n=0}^{2r-1} \frac{1}{\sin\left(\frac{n\pi}{4r}+\frac{\pi}{4}\right)} \\
& = & \frac{2 \sqrt{2}}{r} \sum\limits_{n=0}^{2r-1} \sum\limits_{k=0}^{\infty} \frac{(-1)^{k+1} (2^{2k-1}-1) B_{2k}}{(2k)!} \left(\frac{\pi}{4}\right)^{2k-1} \left( \frac{n}{r}+1 \right)^{2k-1},
\end{eqnarray*}
where, due to $\frac{\pi}{4} \leq (\frac{n\pi}{4r}+\frac{\pi}{4}) < \frac{3 \pi}{4}$ for all $r$, in the last step we used the power expansion of $1/\sin(x) \equiv \csc(x)$ in terms of Bernoulli numbers $B_n$. Splitting off the inner sum the $k=0$ term, and executing the sum over $n$, yields
\begin{eqnarray*}
A(\pi_n)
& = &  \frac{4 \sqrt{2}}{\pi} \sum\limits_{n=0}^{2r-1} \frac{1}{n+r} \\
& + &  \frac{2 \sqrt{2}}{r} \sum\limits_{n=0}^{2r-1} \sum\limits_{k=1}^{\infty} \frac{(-1)^{k+1} (2^{2k-1}-1) B_{2k}}{(2k)!} \left(\frac{\pi}{4}\right)^{2k-1} \left( \frac{n}{r}+1 \right)^{2k-1} \\
& = & \frac{4 \sqrt{2}}{\pi} \big( \Psi(3r) - \Psi(r) \big) \\
& + & 2 \sqrt{2} \sum\limits_{k=1}^{\infty} \frac{(-1)^{k+1} (2^{2k-1}-1) B_{2k}}{(2k)!} \left(\frac{\pi}{4}\right)^{2k-1} \frac{1}{r^{2k}} \\
&   & \hspace*{15mm} \times \big( \zeta(1-2k,r) - \zeta(1-2k,3r) \big),
\end{eqnarray*}
where $\Psi(x)$ denotes the digamma function and $\zeta(n,x)$ the Hurwitz zeta function. Exploiting
\begin{equation*}
\zeta(-n,x) = - \frac{B_{n+1}(x)}{n+1}
\end{equation*}
(see \cite{Apostol95}, Theorem 12.13), which holds for $n \geq 0$ and links the Hurwitz zeta to Bernoulli polynomials 
\begin{equation*}
B_{n}(x) = \sum\limits_{k=0}^n \binom{n}{k} B_{n-k} x^k,
\end{equation*}
we can further simplify $A(\pi_n)$ to
\begin{eqnarray*}
A(\pi_n)
& = & \frac{4 \sqrt{2}}{\pi} \big( \Psi(3r) - \Psi(r) \big) \\
& + & 2 \sqrt{2} \sum\limits_{k=1}^{\infty} \frac{(-1)^{k+1} (2^{2k-1}-1) B_{2k}}{(2k)! \, 2k} \left(\frac{\pi}{4}\right)^{2k-1} \frac{1}{r^{2k}} \big( B_{2k}(3r) - B_{2k}(r) \big).
\end{eqnarray*}
Observing that $B_{n}(x)$ are polynomials of degree $n$ in $x$, and recalling that our assessment aims at the asymptotic limit $r \rightarrow \infty$, the last equation yields
\begin{eqnarray*}
A(\pi_n)
& = & \frac{4 \sqrt{2}}{\pi} \big( \Psi(3r) - \Psi(r) \big) \\
& + & 2 \sqrt{2} \sum\limits_{k=1}^{\infty} \frac{(-1)^{k+1} (2^{2k-1}-1) B_{2k}}{(2k)! \, 2k} \left(\frac{\pi}{4}\right)^{2k-1} ( 3^{2k}-1 ) + \mathcal{O}\left( \tfrac{1}{r} \right).
\end{eqnarray*}
Performing now carefully the asymptotic limit $r \rightarrow \infty$, we finally obtain
\begin{eqnarray}
\label{Eq_pidigi}
\overline{\pi} 
& := & \lim_{r \rightarrow \infty} A(\pi_n) \nonumber \\
& = & \frac{2 \sqrt{2}}{\pi} \left( 2 \ln(3) + \ln(\tfrac{9}{8}) + \ln(8) - 2 \ln(16 (2-\sqrt{2}) + 2 \ln(\tfrac{16}{9} (\sqrt{2}+2)) \right) \nonumber \\
& = & \frac{4 \sqrt{2}}{\pi} \left( \ln(2+\sqrt{2}) - \ln(2-\sqrt{2}) \right) \nonumber \\
& \sim & 3.17406.
\end{eqnarray}
Thus, in the case of the performed parametric discretization of $S^1 \subset \mathbb{R}^2$ given in polar coordinates, the numerical value of $\overline{\pi}$, defined as the arithmetic mean of the $\pi$-values associated with each point along the discretized circle in a space with $\ell^1$-norm, converges to (\ref{Eq_pidigi}), expectedly in accordance with its continuum counterpart (\ref{Eq_picont}). 

We note, however, that the discretization performed above does, in general, not yield points $(x_n,y_n) \in \mathbb{Z}^2$. To ensure the latter, we must replace Eq.~(\ref{Eq_xnyndigi}) with
\begin{equation}
\label{Eq_xnyndigiS}
a_n(r) = \big\lfloor \, |r \cos(\varphi_n)| \, \big\rfloor + \big\lfloor \, |r \sin(\varphi_n)| \, \big\rfloor
\end{equation}
or
\begin{equation}
\label{Eq_xnyndigiR}
a_n(r) = \left\lfloor \, |r \cos(\varphi_n)| + \tfrac{1}{2} \, \right\rfloor + \left\lfloor \, |r \sin(\varphi_n)| + \tfrac{1}{2} \, \right\rfloor,
\end{equation}
where the former ``snaps'' the points along $S^1$ to integer coordinates on $\mathbb{Z}^2$ inside the circle, i.e.
\begin{equation}
\label{Eq_xnynS}
\left\{
\begin{array}{l}
x_n = \lfloor r \cos(\varphi_n) \rfloor \\[0.2em]
y_n = \lfloor r \sin(\varphi_n) \rfloor
\end{array}
\right.
\end{equation}
in the upper right quadrant, whereas the latter associates each point on $S^1$ to the nearest lattice points on $\mathbb{Z}^2$ by rounding independently each coordinate, i.e.
\begin{equation}
\label{Eq_xnynD}
\left\{
\begin{array}{l}
x_n = \lfloor r \cos(\varphi_n) + \frac{1}{2} \rfloor \\[0.2em]
y_n = \lfloor r \sin(\varphi_n) + \frac{1}{2} \rfloor
\end{array}
\right.
\end{equation}
in the upper right quadrant. However, even with these modifications and steps towards a valid discretization, or digital model, of the circle $S^1 \subset \mathbb{R}^2$ in $\mathbb{Z}^2$, the obtained values for $\overline{\pi}$ differ numerically from $\pi$ (see Fig.~\ref{Fig_3}, bottom right).


\subsection{$\pi$ on the digital circle}
\label{SS_piND}

\begin{figure}[t!]
\centering
\includegraphics[width=\linewidth]{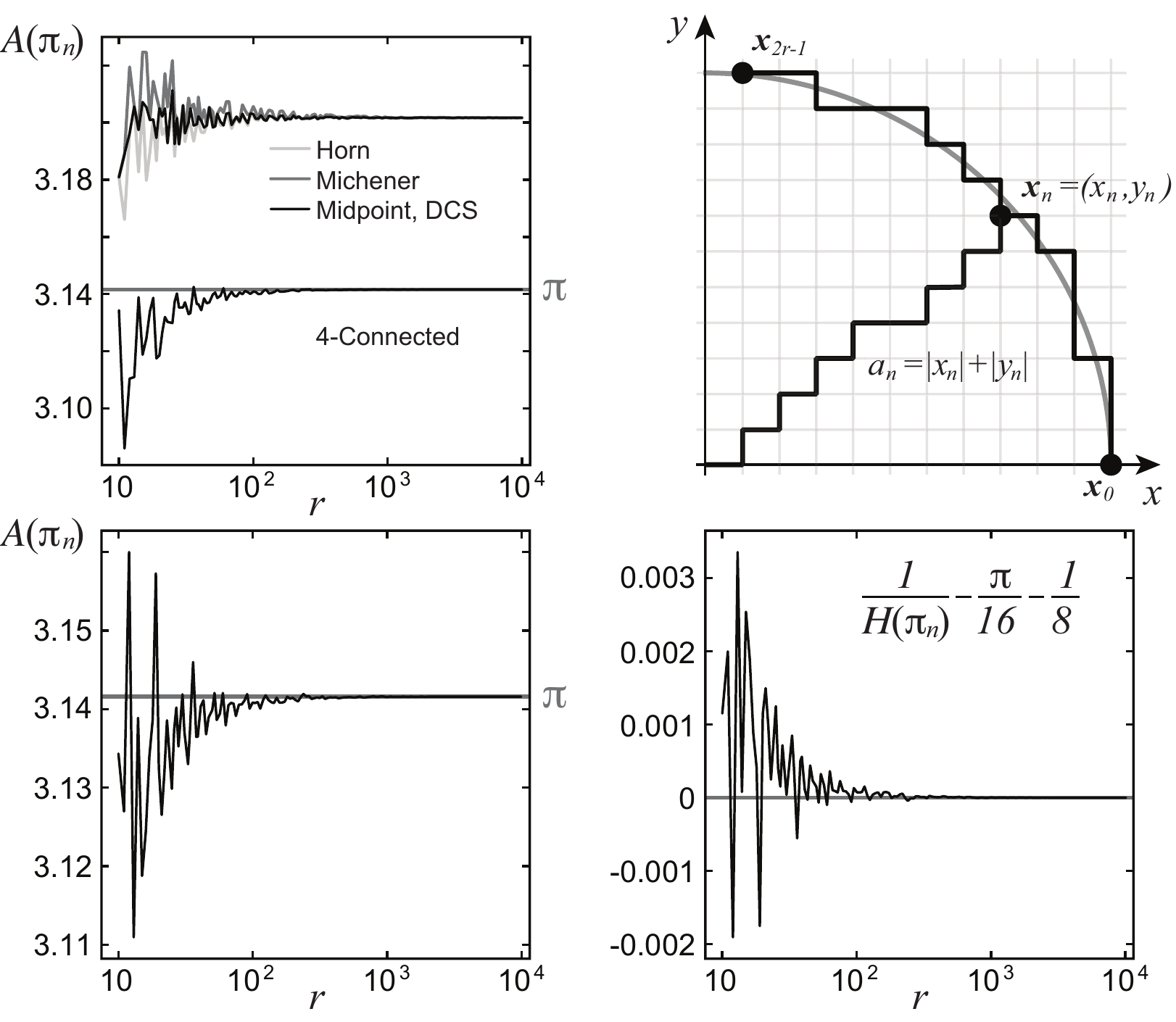}
\caption{\label{Fig_4}
Values of $\overline{\pi}(r)$, defined as the arithmetic mean of all $\pi_n(r)$ associated with each point on a digital circle, as function of $r$ for various digital circle algorithms (top left). Calculation of $\pi_n(r)$ in a digital circle constructed using the signum algorithm (top right; see text for explanation), and resulting $\overline{\pi}(r)$ (bottom left). As for the 4-connected algorithm (see top left), also the signum algorithm yields for large $r$ a numerical value converging to $\pi$ (see Conjecture~\ref{Conj_Api}). Interestingly, considering the harmonic mean of $\pi_n(r)$ yields a value proportional to $\pi$ as well (bottom right; see Proposition~\ref{Prop_Hpi}).
}  
\end{figure}

The above outlined parametric discretization constitutes, in one form or the other, the basis for most published digital circle algorithms. Naturally, values of $\overline{\pi}$, defined as the asymptotic limit of the arithmetic mean of $\pi_n(r)$, see Eq.~(\ref{Eq_ApinND}), will, expectedly, deviate from $\pi$ (see Fig.~\ref{Fig_4}, top left). However, and somewhat surprisingly, this appears to be not true for the 4-connected algorithm and the signum algorithm (Fig.~\ref{Fig_4}, bottom left) introduced here. Focusing on the latter, the numerical assessment of the arithmetic mean of the reciprocal Manhattan distance $a_n = |x_n| + |y_n|$ associated with each recursively generated point $(x_n,y_n) \in \mathbb{Z}^2$ (Fig.~\ref{Fig_4}, top right) according to (\ref{Eq_S1Algorithm}) suggests that, in this case, the correct value for $\pi$ is obtained in the asymptotic limit for $r \rightarrow \infty$ (Fig.~\ref{Fig_4}, bottom left). Specifically, we can formulate the following

\begin{conjecture}
\label{Conj_Api}
The arithmetic mean 
\begin{equation}
A(\pi_n)
= \frac{1}{2r} \sum\limits_{n=0}^{2r-1} \frac{\mathcal{C}}{2a_n(r)} 
= 2 \sum\limits_{n=0}^{2r-1} \frac{1}{a_n(r)},
\end{equation}
of the finite sequence 
\begin{equation}
\pi_n(r) = \frac{\mathcal{C}}{2a_n(r)} = \frac{4r}{a_n(r)},
\end{equation}
where $a_n = |x_n| + |y_n|$ denotes the $\ell^1$-norm of each point $(x_n,y_n) \in \mathbb{Z}^2$ on the digital circle $\mathcal{S}^1 \subset \mathbb{Z}^2$ constructed recursively by (\ref{Eq_S1Algorithm}), converges to $\pi$ in the asymptotic limit $r \rightarrow \infty$, i.e.
\begin{equation}
\lim_{r \rightarrow \infty} A(\pi_n) = \pi.
\end{equation}
\end{conjecture}

The attempt of a rigorous proof of this conjecture can be found in \cite{Rudolph16}. We also note that the same convergence is found in the case of the 4-connected algorithm (\cite{BarreraEA15}; see Fig.~\ref{Fig_4}, left).

Although the recovery of $\pi$ in the case of a valid path describing a digital circle in $\mathbb{Z}^2$, a space with $\ell^1$-norm, is somewhat unexpected, an even more surprising result is obtained when considering the reciprocal of the harmonic mean $H(\pi_n)$, which is proportional to the arithmetic mean of $1/\pi_n \sim a_n(r)$ itself. Specifically, we have

\begin{proposition}
\label{Prop_Hpi}
The harmonic mean 
\begin{equation}
\label{Eq_Hpi}
H(\pi_n) = \left( A\left(\frac{a_n(r)}{4r}\right) \right)^{-1}
\end{equation}
of the sequence of $\pi_n$ values associated with each point $(x_n,y_n) \in \mathbb{Z}^2$ along a digital circle $\mathcal{S}^1 \subset \mathbb{Z}^2$ constructed recursively through (\ref{Eq_S1Algorithm}) obeys, in the asymptotic limit $r \rightarrow \infty$, the identity
\begin{equation}
\label{Eq_Hpi1}
\lim_{r \rightarrow \infty} \frac{1}{H(\pi_n)} = \frac{\pi}{16} + \frac{1}{8}.
\end{equation}
\end{proposition}

\begin{proof}
To show (\ref{Eq_Hpi1}), we first calculate the area $\mathcal{A}(r)$ enclosed by the digital circle $\mathcal{S}^1$ (as above, for notational and symmetry reasons, we will restrict to the quarter circle in the upper right quadrant). To that end, we first construct two associated valid paths $\mathcal{P}_{\text{inner}} \subset \mathbb{Z}^2$ and $\mathcal{P}_{\text{outer}} \subset \mathbb{Z}^2$ by taking the floor and ceiling of each coordinate $(x,y)$ along the circle $S^1 \subset \mathbb{R}^2$. Both paths enclose areas $\mathcal{A}_{\text{inner}}(r)$ and $\mathcal{A}_{\text{outer}}(r)$, respectively (see Fig.~\ref{Fig_5}). By construction, each point along the circular path $\mathcal{S}^1$ will reside inside or on the circumference of $\mathcal{A}_{\text{outer}}(r)$, and outside or on the circumference of $\mathcal{A}_{\text{inner}}(r)$, thus
\begin{equation*}
\mathcal{A}_{\text{inner}}(r) \leq \mathcal{A}(r) \leq \mathcal{A}_{\text{outer}}(r) .
\end{equation*}
Moreover, noting that we consider a quarter circle, and recalling the approximation of the area of a circle $4 \mathcal{A} = \pi r^2$ in $\mathbb{R}^2$ by a Riemannian sum, we have
\begin{equation*}
\mathcal{A}_{\text{inner}}(r) < \frac{1}{4} \pi r^2 < \mathcal{A}_{\text{outer}}(r)
\end{equation*}
with 
\begin{equation*}
\lim_{r \rightarrow \infty} \mathcal{A}_{\text{inner}}(r) 
= \lim_{r \rightarrow \infty} \mathcal{A}_{\text{outer}}(r) 
= \frac{1}{4} \pi r^2,
\end{equation*}
thus
\begin{equation}
\label{Eq_Alimit}
\lim_{r \rightarrow \infty} \mathcal{A}(r) = \frac{1}{4} \pi r^2.
\end{equation}

\begin{figure}[t!]
\centering
\includegraphics[width=\linewidth]{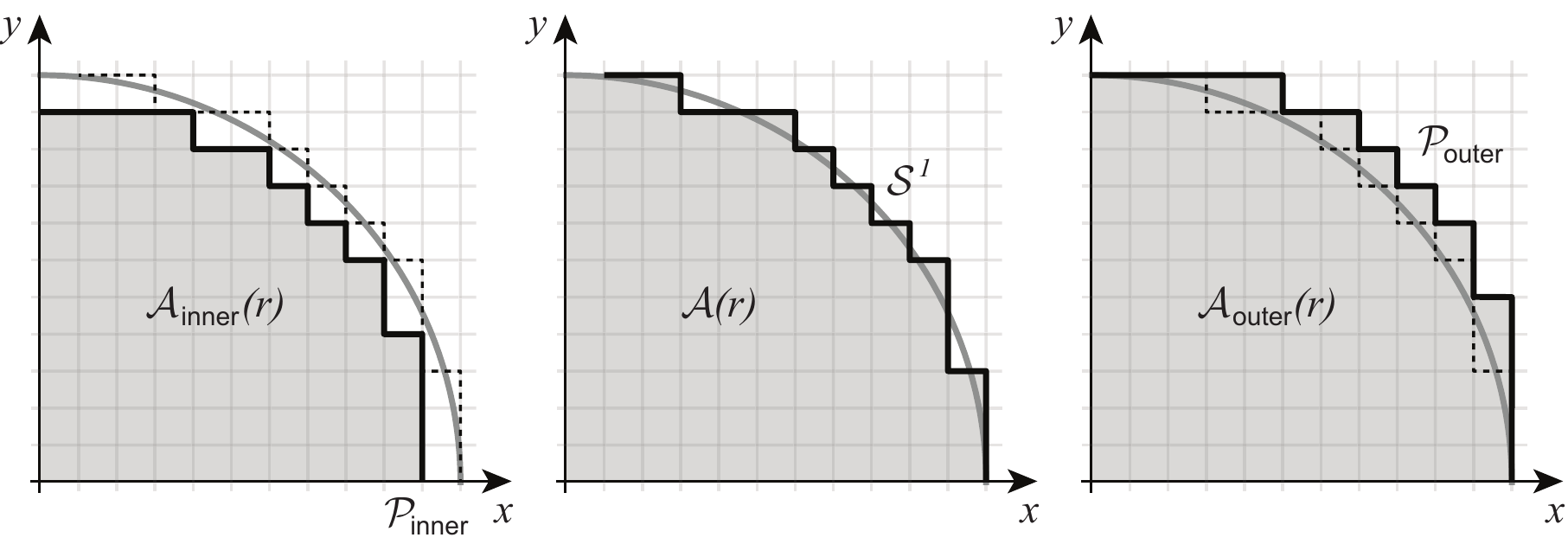}
\caption{\label{Fig_5}
Construction of paths $\mathcal{P}_{\text{inner}} \subset \mathbb{Z}^2$ (left) and $\mathcal{P}_{\text{outer}} \subset \mathbb{Z}^2$ (right) which are enclosed or do enclose the digital circle $\mathcal{S}^1$ (middle), respectively, along with their associated respective areas $\mathcal{A}_{\text{inner}}(r)$, $\mathcal{A}_{\text{outer}}(r)$ and $\mathcal{A}(r)$ in the upper right quadrant (see text for explanation). 
}  
\end{figure}

We next construct recursively the (quarter circle) area enclosed by $\mathcal{S}^1$ through a finite recursive sequence $\mathcal{A}_n(r)$. To that end, we note that $y_{n+1} \neq y_n$ only for $s_n = 1$, whereas $x_{n+1} \neq x_n$ only for $s_n = -1$ (see Proposition~\ref{Prop_SignumAlgorithm}). Starting at $\boldsymbol{x}_0 = (r,0)$, we have 
\begin{equation}
\label{Eq_An}
\mathcal{A}_0 = 0, \mathcal{A}_{n+1} = \mathcal{A}_n + \frac{1}{2} (s_n + 1) x_n
\end{equation}
with $n \in \mathbb{N}, n \in [0, 2r-2]$. If $s_n = 1$, $\mathcal{A}_n$ is updated by the next horizontal ``strip'' according to $\mathcal{A}_{n+1} = \mathcal{A}_n + x_n$, whereas $A_{n+1} = A_n$ in the case of $s_n = -1$. This recursively ``constructs'' the area under $\mathcal{S}^1$ as we go along the circular path $\mathcal{S}^1$. For $n=2r-2$ in (\ref{Eq_An}), we obtain the full area, i.e. 
\begin{equation}
\label{Eq_A}
\mathcal{A}(r) = \mathcal{A}_{2r-1}.
\end{equation}

It remains to evaluate $\mathcal{A}_{2r-1}$. To that end, we first rewrite the recursion (\ref{Eq_An}) in explicit form:
\begin{eqnarray*}
\mathcal{A}_n 
& = & \mathcal{A}_0 + \frac{1}{2} \sum\limits_{k=0}^{n-1} (s_k + 1) x_k \\
& = & \mathcal{A}_0 + \frac{1}{2} \sum\limits_{k=0}^{n-1} s_k x_k + \frac{1}{2} \sum\limits_{k=0}^{n-1} x_k \\
& = & \mathcal{A}_0 + \frac{1}{2} s_0 x_0 + \frac{1}{2} \sum\limits_{k=1}^{n-1} s_k x_k + \frac{1}{2} x_0 + \frac{1}{2} \sum\limits_{k=1}^{n-1} x_k \\
& = & x_0 + \frac{1}{2} \sum\limits_{k=1}^{n-1} s_k x_k + \frac{1}{2} \sum\limits_{k=1}^{n-1} \left( x_0 + \frac{1}{2} S_{k-1} - \frac{k}{2} \right) \\
& = & \frac{1}{2} (n+1) x_0 - \frac{1}{8} n (n-1) + \frac{1}{2} \sum\limits_{k=1}^{n-1} s_k x_k + \frac{1}{4} \sum\limits_{k=1}^{n-1} S_{k-1} ,
\end{eqnarray*}
$n \in [0,2r-1]$, where in the penultimate step we utilized the explicit form of $x_n$,
\begin{equation}
\label{Eq_xnExpl}
x_n = x_0 + \frac{1}{2} S_{n-1} - \frac{n}{2},
\end{equation}
which can easily be deduced from (\ref{Eq_S1Algorithm}) with
\begin{equation}
\label{Eq_Sn}
S_n := \sum\limits_{k=0}^n s_k .
\end{equation}
Applying again (\ref{Eq_xnExpl}), we obtain
\begin{equation*}
\mathcal{A}_n = \frac{1}{8} n ( 1 - n + 4r ) + \frac{1}{4} \sum\limits_{k=1}^{n-1} S_{k-1} + \frac{1}{2} r S_{n-1} + \frac{1}{4} \sum\limits_{k=1}^{n-1} s_k S_{k-1} - \frac{1}{4} \sum\limits_{k=1}^{n-1} k s_k ,
\end{equation*}
where $x_0 = r$ and $s_0 = 1$ were used. This yields, with (\ref{Eq_A}),
\begin{equation}
\label{Eq_An1}
\mathcal{A}(r) 
 =  \frac{1}{4} \left( (r+1)(2r-1) + \sum\limits_{k=1}^{2r-2} S_{k-1} + 2r S_{2r-2} 
 + \sum\limits_{k=1}^{2r-2} s_k S_{k-1} - \sum\limits_{k=1}^{2r-2} k s_k \right).
\end{equation}

We first evaluate $S_{2r-2}$. Due to its definition (\ref{Eq_Sn}), $S_n$ is subject to the recursion
\begin{equation}
\label{Eq_SnRec}
S_0 = s_0 = 1, S_{n+1} = S_n + s_{n+1} 
\end{equation}
with $n \in [0, 2r-2]$, which yields $S_{2r-2} = S_{2r-1} - s_{2r-1}$. Due to symmetry of the lower and upper half of the quarter circle, the number of steps to the left ($s_n = -1$) and upwards ($s_n = 1$) must, by construction, be equal, hence $S_{2r-1} = 0$. Moreover, again due to symmetry, $s_{2r-1} = -1$, which yields 
\begin{equation}
\label{Eq_AnA}
S_{2r-2} = 1. 
\end{equation}

The second last term (\ref{Eq_An}) can be similarly treated, using arguments from symmetry. Specifically, we have
\begin{eqnarray*}
s_n & = & - s_{2r-1-n} \\
S_n & = & S_{2r-1-(n+1)} 
\end{eqnarray*}
$\forall n \in [0,r]$. Thus,
\begin{eqnarray}
\label{Eq_AnB}
\sum\limits_{k=1}^{2r-2} s_k S_{k-1}
& = & \sum\limits_{k=1}^{r-1} s_k S_{k-1} + \sum\limits_{k=r}^{2r-2} s_k S_{k-1} \nonumber \\
& = & \sum\limits_{k=1}^{r-1} s_k S_{k-1} + \sum\limits_{k=1}^{r-1} s_{2r-1-k} S_{2r-1-(k+1)} \nonumber \\
& = & \sum\limits_{k=1}^{r-1} s_k S_{k-1} - \sum\limits_{k=1}^{r-1} s_k S_k \nonumber \\
& = & \sum\limits_{k=1}^{r-1} s_k ( S_{k-1} - S_k ) \nonumber \\
& = & -\sum\limits_{k=1}^{r-1} s_k^2 =  -\sum\limits_{k=1}^{r-1} 1 = -(r-1),
\end{eqnarray}
where in the last step we used again (\ref{Eq_SnRec}) and the fact that $s_n^2 = 1$ for all $n$.

Finally, reordering terms in the last sum in (\ref{Eq_An}) yields
\begin{eqnarray}
\label{Eq_AnC}
\sum\limits_{k=1}^{2r-2} k s_k
& = & \sum\limits_{k=1}^{2r-2} s_k + \sum\limits_{k=2}^{2r-2} s_k + \ldots + \sum\limits_{k=2r-2}^{2r-2} s_k \nonumber \\
& = & (S_{2r-2} - S_0) + (S_{2r-2} - S_1) + \ldots + (S_{2r-2} - S_{2r-3}) \nonumber \\
& = & (2r-2) S_{2r-2} - \sum\limits_{k=0}^{2r-3} S_k \nonumber \\
& = & 2r - 2 - \sum\limits_{k=1}^{2r-2} S_{k-1},
\end{eqnarray}
where in the last step we used (\ref{Eq_AnA}) and changed the summation index in the remaining sum. Taking (\ref{Eq_AnA}), (\ref{Eq_AnB}) and (\ref{Eq_AnC}), we obtain for (\ref{Eq_An})
\begin{equation*}
\mathcal{A}(r) = \frac{1}{2} \left( r^2 + 1 \right) + \frac{1}{2} \sum\limits_{k=1}^{2r-2} S_{k-1},
\end{equation*}
which yields 
\begin{equation}
\label{Eq_sumSn}
\sum\limits_{k=1}^{2r-2} S_{k-1} = 2 \mathcal{A}(r) - r^2 - 1.
\end{equation}

We can now calculate the arithmetic mean of $a_n$, specifically
\begin{eqnarray*}
A\left(\frac{a_n}{4r}\right) 
& = & \frac{1}{2r} \sum\limits_{k=0}^{2r-1} \frac{a_n}{4r} \\
& = & \frac{1}{8r^2} \sum\limits_{k=0}^{2r-1} (r + S_{k-1}).
\end{eqnarray*}
Here we made use of the explicit form of $a_n$, which can easily be deduced from (\ref{Eq_anRec}) as $a_n = a_0 + S_{n-1}$ with $a_0 = r$. Together with (\ref{Eq_sumSn}), we then have
\begin{eqnarray*}
A\left(\frac{a_n}{4r}\right) 
& = & \frac{1}{8r^2} \left( \sum\limits_{k=0}^{2r-1} r + \sum\limits_{k=0}^{2r-1} S_{k-1} \right) \\
& = & \frac{1}{4r^2} \mathcal{A}(r) + \frac{1}{8} - \frac{1}{8r^2},
\end{eqnarray*}
which yields in the asymptotic limit for $r \rightarrow \infty$, using (\ref{Eq_Alimit}), 
\begin{equation*}
\lim_{r \rightarrow \infty} A\left(\frac{a_n}{4r}\right) = \frac{\pi}{16} + \frac{1}{8}.
\end{equation*}
Finally, noting that $\pi_n = \frac{a_n}{4r}$, and that the harmonic mean is the reciprocal dual of the arithmetic mean, we have proven Proposition~\ref{Prop_Hpi}.
\end{proof}


\section{Concluding Remarks}
\label{S_Conclusion}

The results presented in the last section hint at some deeper number-theoretical peculiarities of digital circles, beyond their defining conception as mere digital, or digitized, models of circles in $\mathbb{R}^2$. When considering digital circles rigorously in a discrete space with $\ell^1$-norm, a direct link can be drawn to their continuous ideal $S^1$. We exemplified this point by showing that $\pi$ can be recovered in the asymptotic limit of infinite radius by simply averaging over the $\pi$-values associated with each point along a valid discrete circular path in a space with $\ell^1$-norm (Conjecture~\ref{Conj_Api}). Equally interesting is the finding that also the harmonic mean of this sequence of $\pi$-values yields, in the asymptotic limit, a value  linear in $\pi$ (Proposition~\ref{Prop_Hpi}). Although the fundamental inequality linking the arithmetic and harmonic means of a given sequence is not violated,
\begin{equation}
\lim_{r \rightarrow \infty} A(\pi_n) = \pi > \frac{16}{\pi + 2} = \lim_{r \rightarrow \infty} H(\pi_n),
\end{equation}
the construction of the sequences of $\pi_n$ and their reciprocals suggest an identity linking $\pi$ and its reciprocal.

Finally, the recursive signum algorithm for constructing a valid digital path in $\mathbb{Z}^2$ (Proposition~\ref{Prop_SignumAlgorithm}) approximating $S^1 \subset \mathbb{R}^2$ allows for the construction of a recursive sequence yielding the area inside a circular path, as demonstrated in the proof of Proposition~\ref{Prop_Hpi}. To what extent this approach might be exploitable for gaining deeper insights into the Gauss's Circle Problem remains to be explored.


\section*{Acknowledgments}
Research supported in part by CNRS. The author wishes to thank LE Muller II, J Antolik, D Holstein, JAG Willow, S Hower and OD Little for valuable discussions and comments.



\end{document}